\documentclass{fundam}

\usepackage{amssymb,latexsym,stmaryrd,fancyhdr,graphicx,epic,eepic,rotating,longtable}

\usepackage[polish]{babel}
\usepackage[utf8]{inputenc}

\newcommand{\ignore}[1] { }

\newcommand{\sk}{s}

\newcommand{\fm}{{\mathcal Fm}}

\newcommand{\nmvla}
{\mbox{\boldmath$N \hspace{-0.12 em}
M \hspace{-0.08 em} V \hspace{-0.16 em} L_{\sbA}$}}

\newcommand{\nmvladd}
{\mbox{\boldmath$N \hspace{-0.12 em}
M \hspace{-0.08 em} V \hspace{-0.16 em}
L_{\sbA \sbD \hspace{-0.05 em} \sbD}$}}

\newcommand{\nmvlr}
{\mbox{\boldmath$N \hspace{-0.12 em}
M \hspace{-0.08 em} V \hspace{-0.16 em} L_{\sbR}$}}

\newcommand{\nmvlrdd}
{\mbox{\boldmath$N \hspace{-0.12 em}
M \hspace{-0.08 em} V \hspace{-0.16 em}
L_{\sbR \sbD \hspace{-0.05 em} \sbD}$}}

\newcommand{\nmvlrsd}
{\mbox{\boldmath$N \hspace{-0.12 em}
M \hspace{-0.08 em} V \hspace{-0.16 em}
L_{\sbR \sbS \sbD}$}}

\newcommand{\nmvlrddsd}
{\mbox{\boldmath$N \hspace{-0.12 em}
M \hspace{-0.08 em} V \hspace{-0.16 em}
L_{\sbR \sbD \sbD \sbS \sbD}$}}

\newcommand{\bgamma}{\mbox{\boldmath$\Gamma$}}

\newcommand{\bsigma}{\mbox{\boldmath$\Sigma$}}

\newcommand{\btm}[1]{\mbox{\boldmath \tiny$#1$}}
\newcommand{\bsm}[1]{\mbox{\boldmath \scriptsize$#1$}}
\newcommand{\sbA}{\bsm A}
\newcommand{\sbD}{\bsm D}

\newcommand{\sbR}{\bsm R}
\newcommand{\sbS}{\bsm S}

\newcommand{\sbnmvla}
{{\bsm N} \hspace{-0.19 em}
{\bsm M} \hspace{-0.2 em} {\bsm V} \hspace{-0.2 em} {\bsm L}_{\btm A}}

\newcommand{\sbnmvladd}
{{\bsm N} \hspace{-0.19 em}
{\bsm M} \hspace{-0.2 em} {\bsm V} \hspace{-0.2 em}
{\bsm L}_{{\btm A} {\btm D \hspace{-0.05 em}} {\btm D} }}

\newcommand{\sbnmvlr}
{{\bsm N} \hspace{-0.19 em}
{\bsm M} \hspace{-0.2 em} {\bsm V} \hspace{-0.2 em} {\bsm L}_{\btm R}}

\newcommand{\sbnmvlrdd}
{{\bsm N} \hspace{-0.19 em}
{\bsm M} \hspace{-0.2 em} {\bsm V} \hspace{-0.2 em}
{\bsm L}_{{\btm R} {\btm D} \hspace{-0.05 em} {\btm D} }}

\renewcommand{\psi}{\varphi}

\begin{document}

\setcounter{page}{143}
\publyear{22}
\papernumber{2123}
\volume{186}
\issue{1-4}

      \finalVersionForARXIV

\title{A Note on Calculi for Non-deterministic Many-valued Logics}

\author{Michael Kaminski\thanks{Address of correspondence: Department of Computer Science,
                                Technion -- Israel Institute of Technology, Haifa 3200003, Israel}
   \\
Department of Computer Science\\
Technion -- Israel Institute of Technology\\
Haifa 3200003, Israel\\
kaminski@as.technion.ac.il
}

\runninghead{M. Kaminski}{A Note on Calculi for Non-deterministic Many-valued Logics}

\maketitle

\begin{abstract}
We present two deductively equivalent calculi for
non-deterministic many-valued logics.
One is defined by axioms and the other~-- by rules of inference.
The two calculi are obtained from the truth tables of
the logic under consideration in a straightforward manner.
We prove soundness and strong completeness theorems for both calculi and
also prove the cut elimination theorem for
the calculi defined by rules of inference.
\end{abstract}

\section{Introduction}
\label{s: introduction}

Non-deterministic many-valued logics \cite{AvronK05,AvronL05,AvronZ11}
are a generalization of ``ordinary'' many-valued logics and,
in this note, we extend two of the ``deterministic'' calculi introduced
in~\cite{KaminskiF21} to non-deterministic ones.
Like in~\cite{KaminskiF21},
the logics under considerations are presented semantically,
based on the connectives' truth tables.
The non-deterministic semantics of an $\ell$-ary connective $\ast$
is given by the connective truth table
that is a function from the set of truth values
$V = \{ v_1,\ldots,v_n \}$, $n \geq 2$, into
the set of the non-empty subsets of $V$:
$\ast : V^\ell \rightarrow P(V) \setminus \{ \emptyset \}$.

Similarly to~\cite{KaminskiF21}, we construct proof systems
for non-deterministic many-valued logics out of the truth tables for
the connectives,
cf.~\cite{BaazFZ93,FrancezK19,HanazawaT86,Roussenau67,Roussenau70,
Takahashi68,Takahashi70}.
Our construction is general, transparent, and uniform.

This note is organized as follows.
In Section~\ref{s: nmvla} we introduce
the many-valued logic $\nmvla$\footnote{
The subscript ``$A$'' indicates the axiom description of the logic.}
that is based on
an axiomatic approach and
prove the strong soundness and completeness
(i.e.,  with respect to the consequence relation)
theorem for that logic.
In Section~\ref{s: nmvlr}
we introduce the logic $\nmvlr$ by,
equivalently, replacing some axioms of $\nmvla$ with rules of inference
and
prove the cut elimination theorem.
We conclude the paper with the appendix containing a list of calculi
dual to $\nmvla$ and $\nmvlr$.
The proofs of the properties of these dual calculi are
very similar to their counterparts in~\cite{KaminskiF21} and
are omitted.

\section{Translating truth tables to axioms}
\label{s: nmvla}

The semantics of non-deterministic many-valued logic is as follows.

\medskip
A valuation $v$ is a function from the set of formulas $\fm$ into
the set of truth values $V = \{ v_1,\ldots,v_n \}$, $n \geq 2$,
such that for each connective $\ast$,
\[
v(\ast(\varphi_1,\ldots,\varphi_\ell))
\in
\ast(v(\varphi_1),\ldots,v(\varphi_\ell)))
\]

The logic $\nmvla$
considered in this section has
only structural rules of inference and axioms instead of
logical rules, cf.~\cite{BaazFZ93,Roussenau67}.
We use the notion of
a \textit{labelled} formula that is a pair $(\varphi , k)$,
where $\varphi$ is a formula and $k = 1,\ldots,n$,
introduced in~\cite{BaazLZ13,HanazawaT86}.
The intended meaning of such a labelled formula is that
$v_k$ is the truth value associated with $\varphi$.

Sequents are expressions of the form $\Gamma \rightarrow \Delta$,
where $\Gamma$ and $\Delta$ are finite (possibly empty) sets
of labelled formulas and
$\rightarrow$ is not a symbol of the underlying language.
As we shall see in the sequel, such sequents are more appropriate for
meta-reasoning about labelled formulas than
those from~\cite{BaazFZ93,FrancezK19,Roussenau67,Roussenau70,
Takahashi68,Takahashi70}.

\medskip
The axioms of $\nmvla$ are sequents of the form
\begin{equation}
\label{eq: ordinary axiom}
(\psi , k) \rightarrow (\psi , k)
\end{equation}
$k =1,\ldots,n$, or of the form
\begin{equation}
\label{eq: table axiom}
(\varphi_1 , k_1),\ldots,(\varphi_\ell , k_\ell)
\rightarrow
\{ (\ast(\varphi_1,\ldots,\varphi_\ell),k) :
v_k \in \ast(v_{k_1},\ldots,v_{k_\ell}) \}
\end{equation}
for each table entry $\ast(v_{k_1},\ldots,v_{k_\ell})$.
The latter axiom will be referred to as a \textit{table} axiom.

\medskip
The rules of inference of $\nmvla$ are the structural rules below.
\vspace{0.9 em}

\par \noindent
$k$-L-shift, $k = 1,\ldots,n$,
\begin{equation}
\label{eq: ls}
\frac
{\textstyle
\Gamma , (\varphi , k) \rightarrow \Delta}
{\textstyle
\Gamma
\rightarrow
\Delta , \{ \varphi \} \times \overline{\{ k \}}}
\
\footnote{
As usual,
$\overline{K}$ denotes the complement $\{ 1,\ldots,n \} \setminus K$ of
$K$.}
\end{equation}
\vspace{0.5 em}

\par \noindent
$k^\prime,k^{\prime\prime}$-R-shift,
$k^\prime,k^{\prime\prime} = 1,\ldots,n$,
$k^\prime \neq k^{\prime\prime}$,
\begin{equation}
\label{eq: rs}
\frac
{\textstyle
\Gamma \rightarrow \Delta , (\varphi , k^\prime)}
{\textstyle
\Gamma , (\varphi , k^{\prime\prime}) \rightarrow \Delta}
\end{equation}

\par \noindent
$k$-L-weakening, $k = 1,\ldots,n$,
\begin{equation}
\label{eq: lw}
\frac
{\textstyle
\Gamma \rightarrow \Delta}
{\textstyle
\Gamma , (\varphi ,k) \rightarrow \Delta}
\end{equation}

\par \noindent
$k$-R-weakening, $k = 1,\ldots,n$,
\begin{equation}
\label{eq: rw}
\frac
{\textstyle
\Gamma \rightarrow \Delta}
{\textstyle
\Gamma \rightarrow \Delta , (\varphi ,k)}
\end{equation}

\par \noindent
$k$-cut, $k = 1,\ldots,n$,
\begin{equation}
\label{eq: cut}
\frac
{\textstyle
\Gamma \rightarrow \Delta , (\varphi , k)
\ \ \ \ \ \
\Gamma , (\varphi , k) \rightarrow \Delta}
{\textstyle
\Gamma \rightarrow \Delta}
\end{equation}

\par \noindent
and
\vspace{0.9 em}

\par \noindent
$k^\prime,k^{\prime\prime}$-resolution,
$k^\prime,k^{\prime\prime} =
1,\ldots,n$, $k^\prime \neq k^{\prime\prime}$
\begin{equation}
\label{eq: resolution}
\frac
{\textstyle
\Gamma \rightarrow \Delta^\prime , (\varphi , k^\prime)
\ \ \ \ \ \
\Gamma \rightarrow \Delta^{\prime\prime} , (\varphi , k^{\prime\prime})}
{\textstyle
\Gamma \rightarrow \Delta^\prime , \Delta^{\prime\prime}}
\end{equation}

In fact, by~\cite[Proposition 3.3]{KaminskiF21},
rules~(\ref{eq: cut}) and~(\ref{eq: resolution})
are derivable from each other.

\begin{remark}
\label{r: all calculi}
The axioms~(\ref{eq: ordinary axiom}) belong to
all the calculi considered in this paper and
all the calculi in~\cite{KaminskiF21}.
Also,
the structural rules of all the calculi considered in this paper and
in~\cite{KaminskiF21}
are rules~(\ref{eq: ls})--(\ref{eq: resolution}).
Thus, when the proofs for
``deterministic'' many valued logics in~\cite{KaminskiF21} rely on
axioms~(\ref{eq: ordinary axiom}) and
rules~(\ref{eq: ls})--(\ref{eq: resolution}), only,
they apply to the non-deterministic ones as well.
\end{remark}

\begin{proposition}
\label{p: all phi}
{(\cite[Proposition 3.4]{KaminskiF21})}
The sequent
\begin{equation}
\label{eq: all phi}
\rightarrow \{ \varphi \} \times \{ 1,\ldots,n \}
\end{equation}
is derivable in $\nmvla$.
\end{proposition}

Next, we prove the strong completeness theorem for $\nmvla$.

\begin{definition}
\label{d: satisfiability}
A valuation $v$ \textit{satisfies} a sequent $\Gamma \rightarrow \Delta$
if the following holds.
\begin{itemize}
\item
If for each $(\varphi , k) \in \Gamma$, $v(\varphi) = v_k$,
then for some $(\varphi , k) \in \Delta$, $v(\varphi) = v_k$.\footnote{
\label{f: meta}
That is, $v$ satisfies a sequent $\Gamma \rightarrow \Delta$,
if the meta-value of the classical meta-sequent
$\{ v(\varphi) = v_k : (\varphi,k) \in \Gamma \}
\rightarrow \{ v(\varphi) = v_k : (\varphi,k) \in \Delta \}$
is ``true.''}
\end{itemize}
\end{definition}

\begin{definition}
\label{d: entailment}
A set of sequents $\bsigma$ \textit{semantically entails}
a sequent $\Sigma$,
denoted $\bsigma \models \Sigma$,
if each valuation satisfying all sequents from $\bsigma$
also satisfies $\Sigma$.
\end{definition}

\begin{theorem}
\label{t: c and s of nmvla}
{(Soundness and completeness of $\nmvla$)}
Let $\bsigma$ be a set of sequents.
Then \linebreak $\bsigma \vdash_{\sbnmvla} \Gamma \rightarrow \Delta$
if and only if
$\bsigma \models \Gamma \rightarrow \Delta$.
\end{theorem}

An immediate corollary to Theorem~\ref{t: c and s of nmvla} is that
$\nmvla$ is (strongly) decidable.

Regarding Theorem~\ref{t: c and s of nmvla} itself,
soundness is easy to verify and,
for the proof of completeness, we proceed as follows.

\begin{lemma}
\label{l: disjoint}
{(\cite[Lemma 3.12]{KaminskiF21})}
If ${\not\vdash}_{\sbnmvla} \Gamma \rightarrow \Delta$,
then for no formula $\varphi$ and no $k^\prime,k^{\prime\prime}$ such
that $k^\prime \neq k^{\prime\prime}$,
$(\varphi , k^\prime),(\varphi , k^{\prime\prime}) \in\Gamma$.
\end{lemma}

\noindent \textbf{Proof}
{\bf of the completeness part of Theorem~\ref{t: c and s of nmvla}:}\\
Assume to the contrary that
$\bsigma \ {\not\vdash}_{\sbnmvla} \Gamma \rightarrow \Delta$.
Then, by Zorn's lemma,
there is a maximal (with respect to inclusion) set of labelled formulas
$\bgamma$ including $\Gamma$ such that
for no finite subset $\Gamma^\prime$ of $\bgamma$,
$\bsigma \ {\not\vdash}_{\sbnmvla} \Gamma^\prime \rightarrow \Delta$.

\medskip
We observe that for each formula $\varphi$ there is
a $k \in \{ 1,\ldots,n \}$ such that $(\varphi,k) \in \bgamma$.\footnote{
In other words, $\bgamma$ is \textit{complete},
cf.~\cite[paragraph~3.63]{HanazawaT86} and
the definition of the ``classical'' negation completeness.}
For the proof, assume to the contrary that for
each $k \in \{ 1,\ldots,n \}$
there is a finite subset $\Gamma_k$ of $\bgamma$ such that
\begin{equation}
\label{eq: gamma_k}
\bsigma \vdash_{\sbnmvla} \Gamma_k, (\varphi,k) \rightarrow \Delta
\end{equation}
Then, from~(\ref{eq: all phi}) and~(\ref{eq: gamma_k}),
by $n$ cuts we obtain
\[
\bsigma \vdash_{\sbnmvla} \bigcup\limits_{k=1}^n
\Gamma_k \rightarrow
\Delta
\]
which contradicts the definition of $\bgamma$.

\medskip
Let the valuation $v : \fm \rightarrow \{ v_1,\ldots,v_n \}$
be defined by
\begin{equation}
\label{eq: v}
v(\varphi) = v_k , \ \, \mbox{if} \ \, (\varphi,k) \in \bgamma
\end{equation}

We contend that $v$ is well-defined.

\medskip
First we show that $v$ is a function.
For the proof, assume to the contrary that for some
formula $\varphi$ and some $k^\prime$ and $k^{\prime\prime}$
such that $k^\prime \neq k^{\prime\prime}$ both
$(\varphi,k^\prime)$ and $(\varphi,k^{\prime\prime})$ are in $\bgamma$.
Then, by (the contraposition of) Lemma~\ref{l: disjoint},
\[
\vdash_{\sbnmvla} (\varphi,k^\prime) , (\varphi,k^{\prime\prime})
\rightarrow \Delta
\]
which contradicts the definition of $\bgamma$.

\medskip
Next, we are going to show that the function $v : \fm \rightarrow V$
defined by~(\ref{eq: v}) is indeed a valuation.
\vspace{0.4 em}

The proof is by induction on the complexity of $\varphi$.
The basis (in which $\varphi$ is an atomic formula)
is by the definition of $v$, see~(\ref{eq: v}), and,
for the induction step assume that $\varphi$ is of the form
$\ast(\varphi_1,\ldots,\varphi_\ell)$.

\medskip
Let $(\varphi , k) \in \bgamma$ and
let $v(\varphi_j) = v_{k_j}$, $j = 1,\ldots,\ell$.
By the induction hypothesis,
$(\varphi_j , k_j) \in \bgamma$, $j = 1,\ldots,\ell$.
In addition, from the table axiom~(\ref{eq: table axiom}),
by a number of weakenings, we obtain
\begin{equation}
\label{eq: for v}
(\varphi_1,k_1),\ldots,(\varphi_\ell,k_\ell)
\rightarrow
\Delta ,
\{
(\varphi , k^\prime) : v_{k^\prime} \in \ast(v_{k_1},\ldots,v_{k_\ell})
\}
\end{equation}

\eject

Now, assume to the contrary that
$v(\varphi) \notin \ast(v_{k_1},\ldots,v_{k_j})$.
Then, from~(\ref{eq: for v}) and the axiom
$(\varphi , k) \rightarrow (\varphi , k)$,
by a number of $k,k^\prime$-resolutions,
\[
(\varphi,k), (\varphi_1,k_1),\ldots,(\varphi_\ell,k_\ell)
\rightarrow
\Delta
\]
However, the latter contradicts the definition of $\bgamma$.

\medskip
We note next that for no labelled formula $(\varphi , k) \in \Delta$,
$v(\varphi) = v_k$.
Indeed, if for some $(\varphi , k) \in \Delta$, $v(\varphi) = v_k$,
then, by the definition of $v$, see~(\ref{eq: v}),
$(\varphi , k) \in \bgamma$,
which contradicts the definition of~$\bgamma$.

It remains to prove that $v$ satisfies each sequent in $\bsigma$.
Let $\Gamma^\prime \rightarrow \Delta^\prime \in \bsigma$ be such that
$v$ satisfies each labelled formula in $\Gamma^\prime$,
which, by~(\ref{eq: v}),
is equivalent to $\Gamma^\prime \subseteq \bgamma$.
We have to show that $v$ satisfies some labelled formula in
$\Delta^\prime$,
which, by~(\ref{eq: v}),
is equivalent to $\Delta^\prime \cap \bgamma \neq \emptyset$.

\medskip
Assume to the contrary that $\Delta^\prime \cap \bgamma = \emptyset$.
Let $(\varphi , k) \in \Delta^\prime$ and
let $v(\varphi) = v_{k_\varphi}$.
Then, by~(\ref{eq: v}),
$k_\varphi \neq k$ and $(\varphi, k_\varphi) \in \bgamma$.
From $\Gamma^\prime \rightarrow \Delta^\prime$ by $k,k_\varphi$-R-shifts
(for each $(\varphi , k) \in \Delta^\prime$) we obtain
\[
\bsigma \vdash_{\sbnmvla}
\Gamma^\prime,
\{ (\varphi, k_\varphi) : (\varphi , k) \in \Delta^\prime \} \rightarrow
\]
from which, by weakenings,
\begin{equation}
\label{eq: contradiction}
\bsigma \vdash_{\sbnmvla}
\Gamma^\prime,
\{ (\varphi, k_\varphi) : (\varphi , k) \in \Delta^\prime \} \rightarrow
\Delta
\end{equation}
However,
since
\[
\Gamma^\prime,
\{ (\varphi, k_\varphi) : (\varphi , k) \in \Delta^\prime \} \subseteq
\bgamma
\]
(\ref{eq: contradiction}) contradicts the definition of $\bgamma$.\QED

\section{Replacing axioms with with rules of inference}
\label{s: nmvlr}

The sequent calculus $\nmvlr$ in this section is
the ``sequent counterpart'' of
the deduction system $S \hspace{-0.06 em} F_\mathcal{M}^d$
from~\cite[Section 3.1]{AvronK05}.
Namely, $\nmvlr$ results from $\nmvla$ by replacing
axioms~(\ref{eq: table axiom}) with the rules of inference
\begin{equation}
\label{eq: table rule s}
\frac
{\textstyle
\Gamma \rightarrow
\Delta , (\varphi_j , k_j) ,
\ \ \ \ \ \
j = 1,\ldots,\ell}
{\textstyle
\Gamma \rightarrow \Delta ,
\{ (\ast(\varphi_1,\ldots,\varphi_\ell),k) :
v_k \in \ast(v_{k_1},\ldots,v_{k_\ell}) \}}
\end{equation}
for each table entry $\ast(v_{k_1},\ldots,v_{k_\ell})$.

\begin{proposition}
\label{p: equivalence}
Let $\bsigma$ and $\Sigma$ be a set of sequents and a sequent,
respectively.
Then
$\bsigma \, \vdash_{\sbnmvlr} \Sigma$
if and only if
$\bsigma \, \vdash_{\sbnmvla} \Sigma$.
\end{proposition}

\begin{proof}
The proof is similar to that of~\cite[Proposition 4.1]{KaminskiF21}.
 \eject
For the proof of the ``only if'' part of the proposition
it suffices to show
that axioms~(\ref{eq: table axiom}) are derivable in $\nmvlr$.
The derivation is as follows.
\[
\begin{array}{c}
\begin{array}{c}
(\varphi_j , k_j) \rightarrow (\varphi_j , k_j)
\\
\hline
(\varphi_1 , k_1),\ldots,(\varphi_\ell , k_\ell)
\rightarrow
(\varphi_j , k_j)
\end{array}
\, \mbox{L-weakenings}
\ \ \
\begin{array}{c}
~
\\
j = 1,\ldots,\ell
\end{array}
\\
\hline
(\varphi_1 , k_1),\ldots,(\varphi_\ell , k_\ell)
\!
\hspace{-0.1 em}
\rightarrow
\!
\hspace{-0.1 em}
\{ (\ast(\varphi_1,\ldots,\varphi_\ell),k) :
v_k \in \ast(v_{k_1},\ldots,v_{k_\ell}) \}
\end{array}
\,
\hspace{-0.3 em}
\begin{array}{l}
~
\\
(\ref{eq: table rule s})
\end{array}
\]

Conversely,
for the proof of the ``if'' part of the proposition it suffices to show
that rules~(\ref{eq: table rule s}) are derivable in $\nmvla$.
The derivation is by $\ell$ cuts:
 \scriptsize{\[ \hspace{-0.2 em}
 \frac{\textstyle
\{ \Gamma \rightarrow \Delta, (\varphi_j , k_j) : j = 1, \ldots, \ell \},
\ \ 
(\varphi_1 , k_1), \ldots, (\varphi_\ell , k_\ell)
\rightarrow
\{ (\ast(\varphi_1, \ldots,\varphi_\ell),k) :
v_k \in \ast(v_{k_1}, \ldots, v_{k_\ell}) \}}
{\textstyle
\Gamma
\rightarrow
\Delta ,
\{ (\ast(\varphi_1,\ldots,\varphi_\ell),k) :
v_k \in \ast(v_{k_1}, \ldots, v_{k_\ell}) \}
}
\] }\normalsize

\vspace*{-6mm}
\end{proof}

\begin{corollary}
\label{c: c and s of nmvlr}
{ (Soundness and completeness of $\nmvlr$)}
Let $\bsigma$ and $\Sigma$ be a set of sequents and
a sequent, respectively.
Then $\bsigma \vdash_{\sbnmvlr} \Sigma$
if and only if
$\bsigma \models \Sigma$.
\end{corollary}

\begin{proof}
The corollary follows from Proposition~\ref{p: equivalence} and
Theorem~\ref{t: c and s of nmvla}.
\end{proof}

\begin{theorem}
\label{t: nmvlr c-c elimination}
{(Cut/resolution elimination)}
Each $\nmvlr$-derivable sequent is
derivable without cut or resolution.
\end{theorem}

\begin{proof}
By double induction, the outer on the derivation length and
the inner on the complexity of the principal formula,
we eliminate the first cut/resolution in the derivation.

The outer induction and the basis of the inner induction do not involve
rules~(\ref{eq: table rule s}).
Thus, they are like in the corresponding proofs
in~\cite[Section~4.2]{KaminskiF21}.

The the induction step of the inner induction is treated as follows.

\medskip
Since $\ast(\varphi_1,\ldots,\varphi_\ell)$ may be introduced
into the succeedent, only, this is the case of $k^\prime,k^{\prime\prime}$-res\-ol\-u\-tion~(\ref{eq: resolution}),
with $\varphi$ being $\ast(\varphi_1,\ldots,\varphi_\ell)$.
Namely, we have
\begin{equation}\scriptsize
\label{eq: for resolution}
\hspace{-0.2 em}
\begin{array}{c}
\frac
{\textstyle
\Gamma \rightarrow \Delta , (\varphi_j , k_j^\prime) ,
\ \ \ \ \ \
j = 1,\ldots,\ell}
{\textstyle
\Gamma \rightarrow \Delta ,
\{ (\ast(\varphi_1,\ldots,\varphi_\ell),k) :
v_k \in \ast(v_{k_1^\prime},\ldots,v_{k_\ell^\prime}) \}}
\ \ \ \ \ \
\frac
{\textstyle
\Gamma \rightarrow \Delta , (\varphi_j , k_j^{\prime\prime}) ,
\ \ \ \ \ \
j = 1,\ldots,\ell}
{\textstyle
\Gamma \rightarrow \Delta ,
\{ (\ast(\varphi_1,\ldots,\varphi_\ell),k) :
v_k \in \ast(v_{k_1^{\prime\prime}},\ldots,v_{k_\ell^{\prime\prime}}) \}}
\\
\hline
\Gamma \rightarrow \Delta ,
\{ (\ast(\varphi_1,\ldots,\varphi_\ell),k) :
v_k \in \ast(v_{k_1^\prime},\ldots,v_{k_\ell^\prime}) , \
k \neq k^\prime \}
\cup
\{ (\ast(\varphi_1,\ldots,\varphi_\ell),k) :
v_k \in \ast(v_{k_1^{\prime\prime}},\ldots,v_{k_\ell^{\prime\prime}}) , \
k \neq k^{\prime\prime}\}
\end{array}
\end{equation}

If
$(k_1^\prime,\ldots,k_\ell^\prime)
\neq
(k_1^{\prime\prime},\ldots,k_\ell^{\prime\prime})$
then, for some $j = 1,\ldots,\ell$,
$k_j^\prime \neq k_j^{\prime\prime}$ and
we may apply $k_j^\prime,k_j^{\prime\prime}$-resolution
from which we proceed by R-weakenings:
\[ \scriptsize
\hspace{-0.2 em}
\begin{array}{c}
\frac
{\textstyle
\Gamma \rightarrow
\Delta , (\varphi_j , k_j^\prime )
\ \ \ \ \ \
\Gamma
\rightarrow
\Delta , (\varphi_j , k_j^{\prime\prime} )
}
{\textstyle
\Gamma \rightarrow \Delta
}
\vspace{0.1 em}
\\
\hline
\Gamma \rightarrow \Delta ,
\{ (\ast(\varphi_1,\ldots,\varphi_\ell),k) :
v_k \in \ast(v_{k_1^\prime},\ldots,v_{k_\ell^\prime}) , \
k \neq k^\prime \}
\cup
\{ (\ast(\varphi_1,\ldots,\varphi_\ell),k) :
v_k \in \ast(v_{k_1^{\prime\prime}},\ldots,v_{k_\ell^{\prime\prime}}) , \
k \neq k^{\prime\prime} \}
\end{array}
\]

Otherwise, i.e., if
$(k_1^\prime,\ldots,k_\ell^\prime)
=
(k_1^{\prime\prime},\ldots,k_\ell^{\prime\prime})$
then the conclusion of~(\ref{eq: for resolution}) is
\[
\Gamma \rightarrow \Delta ,
\{ (\ast(\varphi_1,\ldots,\varphi_\ell),k) :
v_k \in \ast(v_{k_1^\prime},\ldots,v_{k_\ell^\prime}) \}
\]
and we can replace~(\ref{eq: for resolution}) with
any of its premises,
 \scriptsize{ \[
\frac
{\textstyle
\Gamma \rightarrow \Delta , (\varphi_j , k_j^\prime) ,
\ \ \ \ \ \
j = 1,\ldots,\ell}
{\textstyle
\Gamma \rightarrow \Delta ,
\{ (\ast(\varphi_1,\ldots,\varphi_\ell),k) :
v_k \in \ast(v_{k_1^\prime},\ldots,v_{k_\ell^\prime}) \}}
\] }\normalsize
say.\footnote{In fact, the case of the equality
$(k_1^\prime,\ldots,k_\ell^\prime)
=
(k_1^{\prime\prime},\ldots,k_\ell^{\prime\prime})$
is the only modification needed for the extension of the proofs
of cut and resolution elimination in~\cite{KaminskiF21}
to the case of non-deterministic logics.}
\end{proof}\vspace*{-4mm}

\appendix

\section{Duality}
\label{s: duality}

In the appendix, we consider two kinds of duality
with respect to $\nmvla$ and $\nmvlr$.
One is the ``distributive'' duality resulting from the duality
between metaconjunction and metadisjunction and
the other is the ``sequent" duality resulting from the duality
between the succeedent and the antecedent of a sequent.

\medskip
Accordingly, Section~\ref{as: dd} deals with
the  ``di-stri-bu-ti-vely''  dual systems $\nmvladd$ and \linebreak $\nmvlrdd$ of
$\nmvla$ and $\nmvlr$, respectively; and
Section~\ref{as: sd} deals with the ``sequent'' duality of
$\nmvlr$ and $\nmvlrdd$.
The results in Section~\ref{as: sd} are
the respective ``antecedent'' counterparts
of those in Sections~\ref{s: nmvlr} and~\ref{as: dd}.

\subsection{Distributive duality}
\label{as: dd}

In this section, we present the calculi $\nmvladd$ and $\nmvlrdd$
which are \textit{distributively dual} and deductively equivalent
to $\nmvla$ and $\nmvlr$, respectively,
cf.~\cite{Roussenau67}.
These calculi are presented in Sections~\ref{sas: dd axioms}
and~\ref{sas: dd rules}.
Both employ the following notation.

\medskip
For a labelled formula $(\ast(\varphi_1,\ldots,\varphi_\ell),k)$,
we define the set of sets of labelled formulas
$(\ast(\varphi_1,\ldots\,$, $\varphi_\ell) \times k)^{-1}$
by
\begin{equation}
\label{eq: d -1}
(\ast(\varphi_1,\ldots,\varphi_\ell) \times k)^{-1} =
\{
\{ (\varphi_1 , k_1),\ldots, (\varphi_\ell , k_\ell) \}
:
v_k \in \ast(v_{k_1},\ldots,v_{k_\ell}) \}
\end{equation}
and,
in what follows, we enumerate the sets in
$(\ast(\varphi_1,\ldots,\varphi_\ell) \times k)^{-1}$
as
\begin{equation}
\label{eq: enumeration}
(\ast(\varphi_1,\ldots,\varphi_\ell),k)^{-1}
=
\{
\{ (\varphi_1 , k_{1,q}),\ldots, (\varphi_\ell , k_{\ell,q}) \}
:
q = 1,\ldots,\sk
\}
\,
\footnote{
Note that $\sk$ depends both on $\ast$ and $k$.}
\end{equation}
That is,
\[
(\ast(\varphi_1,\ldots,\varphi_\ell) \times k)^{-1}
=
\{ \Theta_1,\ldots,\Theta_{\sk} \}
\]
where
\begin{equation}
\label{eq: gamma q}
\Theta_q = \{ (\varphi_1 , k_{1,q}),\ldots, (\varphi_\ell , k_{\ell,q}) \}
\end{equation}
$q = 1,\ldots,\sk$.\footnote{
Note the form of $\Theta_q$:
for each $j = 1,\ldots,\ell$ it contains exactly one labelled formula
with the first component $\varphi_j$.}

\medskip
Next, for sets $\Theta_1,\ldots,\Theta_{\sk}$ of labelled formulas,
$\Theta_q$ as in~(\ref{eq: gamma q}),
$q = 1,\ldots,\sk$,
we define the set $\bigvee\limits_{q=1}^s \Theta_q$
of sets of labelled formulas by

\eject

\hbox{}
\vspace*{-11mm}
\[
\bigvee_{q=1}^s \Theta_q
=
\{
\{ ( \varphi_{j_1} , k_{j_1,1}) ,\ldots, ( \varphi_{j_{\sk}} , k_{j_{\sk},\sk}) \} :
( \varphi_{j_q} , k_{j_q,q}) \in \Theta_q \, , \ q = 1,\dots,\sk \}
\]
That is, for each $q = 1,\ldots,\sk$,
the elements of $\bigvee\limits_{q=1}^s \Theta_q$
contain one element of $\Theta_q$ and nothing more.

Similarly,
for a formula $\ast(\varphi_1,\ldots,\varphi_\ell))$ and
a subset $K$ of $\{1,\ldots,n \}$,
we define the set of sets of labelled formulas
$(\ast(\varphi_1,\ldots,\varphi_\ell) \times K)^{-1}$
by
\[
(\ast(\varphi_1,\ldots,\varphi_\ell) \times K)^{-1} =
\{ \{ (\varphi_1 , k_1),\ldots, (\varphi_\ell , k_\ell) \} : \ast(v_{k_1},\ldots,v_{k_\ell}) = \{ v_k : k \in K \} \}
\]

\begin{remark}
\label{r: notation}
Let $k_1,\ldots,k_\ell \in \{ 1,\ldots , n \}$ and
let $K$ be such that
$\ast(v_{k_1},\ldots,v_{k_\ell}) = \{ v_k : k \in K \}$.
Then, in the above notation, table axioms~(\ref{eq: table axiom}) are
\[
\Theta \rightarrow \{ \ast(\varphi_1,\ldots,\varphi_\ell) \} \times K
\]
for all $\Theta \in (\ast(\varphi_1,\ldots,\varphi_\ell) \times K)^{-1}$.
\end{remark}

\subsubsection{Distributively dual axioms}
\label{sas: dd axioms}

The distributively dual counterpart of $\nmvla$,
denoted $\nmvladd$,
is based on Proposition~\ref{p: dual axiom} below.

\begin{proposition}
\label{p: dual axiom}
{(Cf.~\cite[Proposition~5.3]{KaminskiF21}.)}
For all $k$ and all
\[
\{(\varphi_{j_1} , k_{j_1,1}) ,\ldots, (\varphi_{j_{\sk}} , k_{j_{\sk},\sk})\}
\in \bigvee(\ast(\varphi_1,\ldots,\varphi_\ell) \times k)^{-1}
\]
the sequent
\begin{equation}
\label{eq: dual axiom}
(\ast(\varphi_1, \ldots, \varphi_\ell) , k) \rightarrow
(\varphi_{j_1} , k_{j_1,1}) ,\ldots, (\varphi_{j_{\sk}} , k_{j_{\sk},\sk})
\end{equation}
is derivable in $\nmvla$.
\end{proposition}

The intuition lying behind Proposition~\ref{p: dual axiom}
is as follows.
Using $\bigvee$, $\bigwedge$, $\Longrightarrow$, and $\Longleftrightarrow$
as meta-connectives,
we see that
\begin{eqnarray}
\label{eq: and or}
v_k \in v(\ast(\varphi_1,\ldots,\varphi_\ell))
&
\Longrightarrow
&
\bigvee\limits_{v_k \in \ast(v_{k_1},\ldots,v_{k_\ell})}
\
\bigwedge\limits_{j=1}^\ell v(\varphi_j) = v_{k_j}
\\
\nonumber
&
\Longleftrightarrow
&
\bigwedge
\limits_{\Lambda \in \bigvee (\ast(\varphi_1,\ldots,\varphi_\ell),k)^{-1}}
\
\bigvee\limits_{(\varphi_j,k_j) \in \Lambda} v(\varphi_j) = v_{k_j}
\end{eqnarray}
see~(\ref{eq: d -1}) and~ (\ref{eq: enumeration}).
Now,
writing $(\varphi,k)$ for $v(\varphi) = v_k$,
we rewrite~(\ref{eq: and or}) as
\begin{eqnarray}
\label{eq: and or implication}
\nonumber
(\ast(\varphi_1,\ldots,\varphi_\ell) , k)
&
\Longrightarrow
&
\bigvee\limits_{v_k \in \ast(v_{k_1},\ldots,v_{k_\ell})}
\
\bigwedge\limits_{j=1}^\ell (\varphi_j,k_j)
\\
\label{eq: and or equivalence}
\nonumber
&
\Longleftrightarrow
&
\bigwedge
\limits_{\Lambda \in \bigvee (\ast(\varphi_1,\ldots,\varphi_\ell),k)^{-1}}
\
\bigvee\limits_{(\varphi_j,k_j) \in \Lambda} (\varphi_j,k_j)
\end{eqnarray}

Then, the succeedent of~(\ref{eq: dual axiom}) comes from
rewriting~(\ref{eq: and or}) as conjunction of disjunctions.

\eject

The calculus $\nmvladd$ results from $\nmvla$ by replacing
axioms~(\ref{eq: table axiom})
with the (distributively dual) axioms~(\ref{eq: dual axiom}),
cf.~\cite{Roussenau67}.

\begin{theorem}
\label{t: c and s of nmvladd}
{(Soundness and completeness of $\nmvladd$,
cf.~\cite[Theorem~5.5]{KaminskiF21}.)}
Let $\bsigma$ and $\Sigma$ be a set of sequents and
a sequent, respectively.
Then $\bsigma \vdash_{\sbnmvladd} \Sigma$
if and only if
$\bsigma \models \Sigma$.
\end{theorem}

\subsubsection{Distributively dual rules of inference}
\label{sas: dd rules}

The distributively dual calculus $\nmvlrdd$ results from $\nmvlr$
by replacing rules of inference~(\ref{eq: table rule s})
with the rule
\begin{equation}
\label{eq: table rule sp}
\frac
{\textstyle
\Gamma \rightarrow \Delta , \Lambda , \
\Lambda \in \bigvee (\ast(\varphi_1,\ldots,\varphi_\ell),K)^{-1}
}
{\textstyle
\Gamma \rightarrow \Delta ,
\{ \ast(\varphi_1,\ldots,\varphi_\ell) \} \times K }
\end{equation}

\begin{theorem}
\label{t: c and s of nmvlrdd}
{(Soundness and completeness of $\nmvlrdd$,
cf.~\cite[Theorem~5.7]{KaminskiF21}.)}
Let $\bsigma$ and $\Sigma$ be a set of sequents and
a sequent, respectively.
Then $\bsigma \vdash_{\sbnmvlrdd} \Sigma$
if and only if
$\bsigma \models \Sigma$.
\end{theorem}

\begin{theorem}
\label{t: nmvlrdd c-c elimination}
{(Cf.~\cite[Theorem~5.8]{KaminskiF21}.)}
Each $\nmvlrdd$-derivable sequent is
derivable without cut or resolution.
\end{theorem}

\subsection{Sequent duality}
\label{as: sd}

One can think of $\nmvlr$ and $\nmvlrdd$ as sequent calculi with
rules of introduction to the succeedent only.
In this section we replace these rules with
the dual rules of introduction to the antecedent and
show the dual calculi posses all properties of the original ones.

\medskip
For cut/resolution elimination in the calculi
in this section we need the following rule of inference.

\medskip
\par \noindent
$K$-L-multi-shift, $K \subset \{ 1,\ldots,n \}$,
\begin{equation}
\label{eq: jls}
\frac
{\textstyle
\Gamma , (\varphi , k) \rightarrow \Delta , \ k \in K}
{\textstyle
\Gamma \rightarrow \Delta , \{ \varphi \} \times \overline{K}}
\end{equation}

\begin{remark}
\label{r: m shift}
This rule is derivable from its premises and~(\ref{eq: all phi}) by cuts.
However, we do not call it ``cut,''
because it does not affect the subformula property\footnote{
\label{f: subformula}
We say that a labelled formula $(\varphi,k)$ is a subformula
of a labelled formula $(\varphi^\prime , k^\prime)$,
if $\varphi$ is a subformula of $\varphi^\prime$.}:
since $K$ is a proper subset of $\{ 1,\ldots,n \}$,
$\overline{K} \neq \emptyset$.
\end{remark}

\subsubsection{The sequent dual of $\nmvlr$}
\label{sas: nmvlrsd}

The (dual) rules for introduction of $\ast$ to the antecedent are
\begin{equation}
\label{eq: table rule a}
\frac
{\textstyle
\Gamma ,(\varphi_1 , k_1),\ldots,(\varphi_\ell , k_\ell)
\rightarrow \Delta
, \
v_k \in \ast(v_{k_1},\ldots,v_{k_\ell})}
{\textstyle
\Gamma , (\ast(\varphi_1,\ldots,\varphi_\ell) , k) \rightarrow \Delta}
\end{equation}
for each $k = 1,\ldots,n$,
cf.~\cite[Proof of Theorem~8]{Kaminski88}.

\medskip
The proofs of Propositions~\ref{p: antecedent}
and~\ref{p: antecedent dual axiom}
and
Corollary~\ref{c: r rsd equivalence} below are like
the corresponding proofs of~\cite[Proposition~5.11]{KaminskiF21},
\cite[Proposition~5.12]{KaminskiF21}),
and~\cite[Corollary 5.13]{KaminskiF21}.

\begin{proposition}
\label{p: antecedent}
Rules~{ (\ref{eq: table rule a})} are derivable in $\nmvlr$.
\end{proposition}

\begin{proposition}
\label{p: antecedent dual axiom}
Axioms~{(\ref{eq: dual axiom})} are derivable
by~{\em (\ref{eq: table rule a})}.
\end{proposition}

The sequent calculus $\nmvlrsd$ results from $\nmvlr$ by replacing
rules of introduction to succeedent~(\ref{eq: table rule s})
with rules of introduction to antecedent (\ref{eq: table rule a}) and
adding to it the $K$-L-multi-shift (\ref{eq: jls}).

\begin{corollary}
\label{c: r rsd equivalence}
Calculi $\nmvlrsd$ and $\nmvlr$ are deductively equivalent.
\end{corollary}

\begin{theorem}
\label{t: nmvlrsd c-c elimination}
{(Cf.~\cite[Theorem~5.14]{KaminskiF21}.)}
Each $\nmvlrsd$-derivable sequent is
derivable without cut or resolution.
\end{theorem}

\subsubsection{The sequent dual of $\nmvlrdd$}
\label{sas: nmvlrddsd}

In the case of $\nmvlrdd$,
the (dual) rules for introduction of $\ast$ to the antecedent are
\begin{equation}
\label{eq: table rule dual a}
\frac
{\textstyle
\Gamma ,(\varphi_j , k_j) \rightarrow \Delta
, \
(\varphi_j , k_j) \in \Lambda}
{\textstyle
\Gamma , (\ast(\varphi_1,\ldots,\varphi_\ell) , k) \rightarrow \Delta}
\end{equation}
for each
\vspace{-1.2 em}

\[
\Lambda =
(\varphi_{j_1} , k_{j_1,1}) ,\ldots, (\varphi_{j_{\sk}} , k_{j_{\sk},\sk})
\in \bigvee (\ast(\varphi_1,\ldots,\varphi_\ell),k)^{-1}
\]

The proofs of Propositions~\ref{p: dual antecedent}
and~\ref{p: dual antecedent dual axiom}
and
Corollary~\ref{c: r rddsd equivalence} below are like in
the corresponding proofs of~\cite[Proposition~5.17]{KaminskiF21},
\cite[Proposition~5.18]{KaminskiF21})
and~\cite[Corollary 5.19]{KaminskiF21}.

\begin{proposition}
\label{p: dual antecedent}
Rules~{(\ref{eq: table rule dual a})} are derivable in $\nmvlrdd$.
\end{proposition}

\begin{proposition}
\label{p: dual antecedent dual axiom}
Axioms~{(\ref{eq: dual axiom})} are derivable
from~{\em (\ref{eq: table rule dual a})}.
\end{proposition}

The sequent calculus $\nmvlrddsd$ results from $\nmvlrdd$ by replacing
rules of introduction to succeedent~(\ref{eq: table rule sp})
with rules of
introduction to antecedent~(\ref{eq: table rule dual a}) and
adding to it multi-shifts (\ref{eq: jls}).

\begin{corollary}
\label{c: r rddsd equivalence}
Calculi $\nmvlrddsd$ and $\nmvlrdd$ are deductively equivalent.
\end{corollary}

\begin{theorem}
\label{t: nmvlrddsd c-c elimination}
{(Cf.~\cite[Theorem~5.20]{KaminskiF21}.)}
Each $\nmvlrddsd$-derivable sequent is
derivable without cut or resolution.
\end{theorem}



\begin{thebibliography}{10}
\providecommand{\url}[1]{\texttt{#1}}
\providecommand{\urlprefix}{URL }
\expandafter\ifx\csname urlstyle\endcsname\relax
  \providecommand{\doi}[1]{doi:\discretionary{}{}{}#1}\else
  \providecommand{\doi}{doi:\discretionary{}{}{}\begingroup
  \urlstyle{rm}\Url}\fi
\providecommand{\eprint}[2][]{\url{#2}}


\bibitem{AvronK05}
Avron A, and Konikowska B.
\newblock Multi-valued calculi for logics based on non-determinism.
\newblock {\em Logic Journal of {IGPL}}, 2005. 13(4):365--387.
doi:10.1093/jigpal/jzi030.

\bibitem{AvronL05}
Avron A, and Lev I.
\newblock Non-deterministic multiple-valued structures.
\newblock {\em Journal of Logic and Computation}, 2005. 15(3):241--261.
doi:10.1093/logcom/exi001.

\bibitem{AvronZ11}
Avron A, and Zamansky A.
\newblock Non-deterministic semantics for logical systems.
\newblock In D.M. Gabbay and F.~Guenthner, editors, {\em Handbook of
  Philosophical Logic}, Springer 2011 pages 227--304.
  doi:10.1007/978-94-007-0479-4\_4.

\bibitem{BaazFZ93}
Baaz M,  Ferm{\"{u}}ller CG, and Zach R.
\newblock Systematic construction of natural deduction systems for many-valued
  logics.
\newblock In {\em Proceedings of the 23rd {IEEE} International Symposium on
  Multiple-Valued Logic, {ISMVL}}, pages 208--213. {IEEE} Computer Society,
  1993. doi:10.1109/ISMVL.1993.289558.

\bibitem{BaazLZ13}
Baaz M, Lahav O, and Zamansky A.
\newblock Finite-valued semantics for canonical labelled calculi.
\newblock {\em Journal of Automated Reasoning}, 2013. 51:401--430.
doi:10.1007/s10817-013-9273-x.

\bibitem{FrancezK19}
Francez N, and Kaminski M.
\newblock On poly-logistic natural-deduction for finitely-valued propositional
  logics.
\newblock {\em Journal of Applied Logic}, 2019. 6:255--289.
doi:10.1007/s11787-019-00219-z.

\bibitem{HanazawaT86}
Hanazawa M, and Takano M.
\newblock On intuitionistic many-valued logics.
\newblock {\em Journal of the Mathematical Society of Japan}, 1986. 38:409--419.


\bibitem{Kaminski88}
Kaminski M.
\newblock Nonstandard connectives of intuitionistic propositional logic.
\newblock {\em Notre Dame Journal of Formal Logic}, 1988. 29:309--331.
doi:10.1305/ndjfl/1093637931.

\bibitem{KaminskiF21}
Kaminski M, and Francez N.
\newblock Calculi for many-valued logics.
\newblock {\em Logica Universalis}, 2021. 15:193--226.
doi:10.1007/s11787-021-00274-5.

\bibitem{Roussenau67}
Roussenau G.
\newblock Sequents in many valued logic {I}.
\newblock {\em Fundamenta Mathematicae}, 1967. 60:23--33.
doi:10.4064/ FM-67-1-125-131.

\bibitem{Roussenau70}
Roussenau G.
\newblock Sequents in many valued logic {II}.
\newblock {\em Fundamenta Mathematicae}, 67:125--131, 1970.
 doi:10.4064/ fm-67-1-125-131.

\bibitem{Takahashi68}
Takahashi M.
\newblock Many-valued logics of extended {G}entzen style {I}.
\newblock {\em Science Reports of the Tokyo Kyoiku Daigaku, Section A},
  1968. 9:271--292. URL \url{https://www.jstor.org/stable/43699119}.

\bibitem{Takahashi70}
Takahashi M.
\newblock Many-valued logics of extended {G}entzen style {II}.
\newblock {\em Journal of Symbolic Logic}, 1970. 35:493--528.
doi:10.2307/2271438.
\end{thebibliography}
\end{document}